\documentclass[10pt]{article}
\usepackage[T1]{fontenc}
\usepackage[utf8]{inputenc}
\usepackage{lmodern}
\usepackage{microtype}
\usepackage{fullpage}
\usepackage{authblk}
\usepackage{amsmath,amssymb,amsfonts,amsthm,mathtools}
\usepackage{aliascnt}
\usepackage{mathrsfs}
\usepackage{bm}
\usepackage{dsfont}
\usepackage{mathpazo}
\usepackage{booktabs}
\usepackage{tabularx}
\usepackage{longtable}
\usepackage{array}
\usepackage{xcolor}
\usepackage[shortlabels]{enumitem}
\usepackage{mdframed}
\usepackage{float}
\usepackage{url}
\usepackage{qcircuit}
\usepackage{optidef}
\usepackage[numbers,sort&compress]{natbib}
\usepackage{hyperref}
\usepackage[nameinlink,noabbrev]{cleveref}

\definecolor{MainBlue}{HTML}{173B57}
\definecolor{AccentTeal}{HTML}{247B7B}
\definecolor{SoftBlue}{HTML}{F1F6FA}
\definecolor{SoftTeal}{HTML}{EFF8F7}
\definecolor{SoftGray}{HTML}{F5F6F7}
\definecolor{RuleGray}{HTML}{D8DEE4}
\definecolor{WarmGold}{HTML}{A66F18}

\hypersetup{
colorlinks=true,
linkcolor=blue!55!black,
citecolor=blue!55!black,
urlcolor=blue!55!black,
pdfborder={0 0 0}
}
\allowdisplaybreaks
\theoremstyle{plain}
\newtheorem{theorem}{Theorem}[section]
\newaliascnt{proposition}{theorem}

\aliascntresetthe{proposition}
\newaliascnt{lemma}{theorem}
\newtheorem{lemma}[lemma]{Lemma}
\aliascntresetthe{lemma}
\newaliascnt{corollary}{theorem}
\newtheorem{corollary}[corollary]{Corollary}
\aliascntresetthe{corollary}
\theoremstyle{definition}
\newaliascnt{definition}{theorem}

\aliascntresetthe{definition}
\newaliascnt{example}{theorem}

\aliascntresetthe{example}
\newaliascnt{problem}{theorem}

\aliascntresetthe{problem}
\theoremstyle{remark}
\newaliascnt{remark}{theorem}

\aliascntresetthe{remark}

\mdfdefinestyle{theoremframe}{
backgroundcolor=SoftBlue,
linecolor=MainBlue,
linewidth=.65pt,
innerleftmargin=7pt,innerrightmargin=7pt,
innertopmargin=6pt,innerbottommargin=6pt,
needspace=6\baselineskip,
skipabove=9pt,skipbelow=9pt
}
\mdfdefinestyle{propositionframe}{
backgroundcolor=SoftTeal,
linecolor=AccentTeal,
linewidth=.6pt,
innerleftmargin=7pt,innerrightmargin=7pt,
innertopmargin=6pt,innerbottommargin=6pt,
needspace=6\baselineskip,
nobreak=true,
skipabove=8pt,skipbelow=8pt
}
\mdfdefinestyle{lemmaframe}{
backgroundcolor=SoftGray,
linecolor=RuleGray!60!black,
linewidth=.5pt,
innerleftmargin=7pt,innerrightmargin=7pt,
innertopmargin=5pt,innerbottommargin=5pt,
needspace=5\baselineskip,
nobreak=true,
skipabove=8pt,skipbelow=8pt
}
\surroundwithmdframed[style=theoremframe]{theorem}
\surroundwithmdframed[style=propositionframe]{proposition}
\surroundwithmdframed[style=propositionframe]{corollary}
\surroundwithmdframed[style=lemmaframe]{lemma}

\crefname{theorem}{theorem}{theorems}
\Crefname{theorem}{Theorem}{Theorems}
\crefname{proposition}{proposition}{propositions}
\Crefname{proposition}{Proposition}{Propositions}
\crefname{lemma}{lemma}{lemmas}
\Crefname{lemma}{Lemma}{Lemmas}
\crefname{corollary}{corollary}{corollaries}
\Crefname{corollary}{Corollary}{Corollaries}
\crefname{definition}{definition}{definitions}
\Crefname{definition}{Definition}{Definitions}
\crefname{example}{example}{examples}
\Crefname{example}{Example}{Examples}
\crefname{problem}{problem}{problems}
\Crefname{problem}{Problem}{Problems}
\crefname{remark}{remark}{remarks}
\Crefname{remark}{Remark}{Remarks}
\crefname{equation}{equation}{equations}
\Crefname{equation}{Equation}{Equations}

\newcommand{\ket}[1]{|#1\rangle}
\newcommand{\bra}[1]{\langle#1|}

\newcommand{\ketbra}[2]{|#1\rangle\!\langle#2|}

\newcommand{\tr}{\operatorname{Tr}}

\newcommand{\idop}{\mathds{1}}

\newcommand{\supp}{\operatorname{supp}}

\newcommand{\CC}{\mathbb{C}}

\newcommand{\cH}{\mathcal{H}}

\newcommand{\cL}{\mathcal{L}}

\newcommand{\ox}{\otimes}

\newcommand{\eqt}[1]{\stackrel{\mathclap{\mbox{\scriptsize #1}}}{=}}
\newcommand{\leqt}[1]{\stackrel{\mathclap{\mbox{\scriptsize #1}}}{\leq}}

\newcommand{\lt}[1]{\stackrel{\mathclap{\mbox{\scriptsize #1}}}{<}}

\DeclareMathOperator{\im}{Im}

\numberwithin{equation}{section}

\begin{document}

\title{The Min-Rains Relative Entropy Is Not Tight for Exact PPT Entanglement Distillation}
\author[1,*]{Chengkai Zhu}
\author[2,$\dagger$]{Xin Wang}
\affil[1]{QudeLeap Research, Shanghai 200030, China}
\affil[2]{Thrust of Artificial Intelligence, Information Hub,\par
The Hong Kong University of Science and Technology (Guangzhou), Guangzhou 511453, China}
\date{\today}

\begingroup
\renewcommand{\thefootnote}{\fnsymbol{footnote}}

\footnotetext[1]{\href{mailto:zhuchengkai7@gmail.com}
{zhuchengkai7@gmail.com}}

\footnotetext[2]{\href{mailto:felixxinwang@hkust-gz.edu.cn}
{felixxinwang@hkust-gz.edu.cn}}

\endgroup
\maketitle

\begin{abstract}
Exact entanglement distillation converts a noisy bipartite state into a maximally entangled state with zero error. Under completely PPT-preserving operations, the additive min-Rains relative entropy provides a single-letter upper bound on the regularized distillation rate. An interesting problem in entanglement theory dating back to 2016 has been whether this bound is always tight. Here we resolve this question in the negative. The key is a tensor-stable rigidity absent from the min-Rains relaxation: every feasible exact-distillation effect must act as the identity on the support of the input state. We convert this constraint into a new single-letter upper bound using a generally non-Hermitian, range-supported witness. For a rank-three subspace, we show this upper bound lies strictly below the min-Rains relative entropy for every state with that support. Thus, the constraint discarded by the min-Rains relaxation remains relevant under arbitrary tensor powers. Our result rules out the min-Rains relative entropy as a closed-form formula for exact PPT distillable entanglement and reveals the subtle asymptotic structure of exact entanglement manipulation under PPT operations.
\end{abstract}


\section{Introduction}
\label{sec:introduction}

Entanglement distillation converts noisy bipartite states into maximally
entangled pairs. Such pairs enable quantum
teleportation~\cite{Bennett1993Teleportation}, entanglement-based
cryptography~\cite{Ekert1991}, quantum repeaters~\cite{Briegel1998}, and
distributed quantum computation~\cite{Cirac1999Distributed}, making
distillation a central task in quantum information theory~\cite{Khatri2025a,Watrous2018,Hayashi2017b}.
In its usual asymptotic formulation, the conversion error may vanish with increasing blocklength~\cite{Bennett1996Purification,Bennett1996Mixed,Horodecki2009,Fang2019}.
Exact distillation imposes a qualitatively different requirement where the output must be exactly maximally entangled at every finite blocklength~\cite{Duan2005,Matthews2008}. Under this requirement, the nonzero eigenvalues of the input state are irrelevant, and the task is governed by the geometry of its support~\cite{Wang2016,Regula2019OneShot}.

Completely PPT-preserving operations provide a tractable relaxation of LOCC~\cite{Rains1999,Rains2001,ChitambarGour2019}.  For a state with support projection $P$, the one-shot exact PPT distillation is characterized by~\cite{Wang2016}
\begin{equation}
W_0(P)=\min\{\|E^\Gamma\|_\infty:P\le E\le\idop\}.
\end{equation}
Dropping the upper constraint gives the multiplicative program
\begin{equation}
M(P)=\min\{\|R^\Gamma\|_\infty:P\le R\}
\end{equation}
and the min-Rains relative entropy $R_{\min}(\rho)=-\log M(P)$~\cite{Wang2016c,XWthesis,Baeuml2019}.  Consequently,
$R_{\min}$ is an additive single-letter upper bound on the regularized exact PPT distillable entanglement.  The 2017 work that introduced this bound left open whether it is tight for every state~\cite{Wang2016c}, i.e., whether
\begin{equation}
    E_{0,\mathrm{PPT}}^\infty(\rho)\stackrel{?}{=}R_{\min}(\rho).
\end{equation}
The asymptotic tightness question was posed explicitly again in~\cite[Section~3.4]{Baeuml2019}.  Equivalently, one asks whether the discarded constraint $E\le\idop$ becomes asymptotically negligible under tensor powers.

We prove that it need not.  The inequalities $P\le E\le\idop$ force every
feasible exact-distillation effect to act as the identity on $\im(P)$ and to
have no support--complement coherences.  A normalized, generally
non-Hermitian operator $T$ satisfying $PT=T$ detects this fixed action and
tensorizes into a regularized bound.  We construct an explicit rank-three
qutrit--qutrit support and exact certificates $T$ and $R\ge P$ satisfying
$\|T^\Gamma\|_1\|R^\Gamma\|_\infty<1$.  Hence every state with this support
obeys
\begin{align}
E_{0,\mathrm{PPT}}^\infty(\rho)<R_{\min}(\rho).
\label{eq:intro-strict-separation}
\end{align}

\section{Exact Entanglement Distillation Under PPT Operations}
\label{sec:preliminaries}
All Hilbert spaces are finite-dimensional, and all logarithms are to base two.  We write $\cL(\cH)$ for the space of linear operators on $\cH$.  For $X\in\cL(\cH)$, the trace norm is $\|X\|_1=\tr\sqrt{X^\dagger X}$, and $\|X\|_\infty$ is its largest singular value.  A quantum state is an operator $\rho\ge0$ with $\tr\rho=1$.  For a bipartite system $\cH_{AB}=\cH_A\ox\cH_B$, partial transpose on $B$ is denoted by $\Gamma$, i.e., $(\ketbra{i}{k}\ox\ketbra{j}{\ell})^\Gamma=\ketbra{i}{k}\ox\ketbra{\ell}{j}$.  For a state $\rho$, let $P=\Pi_{\supp(\rho)}$ be the orthogonal projection onto its support.  For an integer $M\ge1$, let $\cH_{A'}\simeq\cH_{B'}\simeq\CC^M$ and define $\ket{\Phi_M}\coloneqq 1/\sqrt{M}\sum_{j=1}^{M}\ket{j}_{A'}\ket{j}_{B'},~\Phi_M\coloneqq\ket{\Phi_M}\!\bra{\Phi_M}$.
A quantum channel is a completely positive, trace-preserving (CPTP) linear map.  It is a \emph{PPT operation} if it is completely PPT-preserving \cite{Rains1999r,Rains1999}.  Equivalently, a channel
$\mathcal N_{AB\to A'B'}$ is PPT when $\Gamma_{B'}\circ\mathcal N\circ\Gamma_B$ is completely positive.

The one-shot exact PPT distillable entanglement is defined by~\cite{Wang2016,Fang2019}
\begin{align}
\widehat E_{0,\mathrm{PPT}}^{(1)}(\rho)
\coloneqq
\max\Bigl\{\log M:\;&M\in\mathbb N_{\ge1},\ 
\exists\,\mathcal N_{AB\to A'B'} \text{ PPT operation, such that }\mathcal N(\rho)=\Phi_M\Bigr\}.
\label{eq:operational-one-shot-rate}
\end{align}
Wang and Duan showed that this quantity depends only on the support
projection $P$ and admits an SDP characterization~\cite{Wang2016}.  Define
\begin{align}
W_0(P)
\coloneqq
\min\bigl\{\|E^\Gamma\|_\infty:P\le E\le\idop\bigr\}.
\label{eq:W0-program}
\end{align}
Exact distillation to $\Phi_M$ is possible precisely when
$M\le W_0(P)^{-1}$.  Restricting the output dimension to integers therefore
gives $\widehat E_{0,\mathrm{PPT}}^{(1)}(\rho)
=\log\left\lfloor W_0(P)^{-1}\right\rfloor$.  For tensor-product arguments,
it is convenient to remove this integer rounding and introduce the
continuous one-shot quantity
\begin{align}
E_{0,\mathrm{PPT}}^{(1)}(\rho)
\coloneqq-\log W_0(P).
\label{eq:one-copy-rate}
\end{align}
The feasible choice $E=\idop$ gives $W_0(P)\le1$.  Hence
$x\coloneqq W_0(P)^{-1}\ge1$.  Since
$x/2\le\lfloor x\rfloor\le x$, it follows that $E_{0,\mathrm{PPT}}^{(1)}(\rho)-1\le\widehat E_{0,\mathrm{PPT}}^{(1)}(\rho)\le E_{0,\mathrm{PPT}}^{(1)}(\rho)$.
The regularized exact PPT distillable entanglement is then given by
\begin{align}
E_{0,\mathrm{PPT}}^\infty(\rho)\coloneqq
\sup_{n\ge1}\frac1n
\widehat E_{0,\mathrm{PPT}}^{(1)}(\rho^{\ox n})=\lim_{n\to\infty}\frac1n
E_{0,\mathrm{PPT}}^{(1)}(\rho^{\ox n}),
\label{eq:regularized-zero-error-rate}
\end{align}
where the equality follows from Lemma~\ref{lem:regularized-rate-exists}.  We abbreviate $W_n(P)\coloneqq W_0(P^{\ox n})$.

The single-letter bound used below is best understood as part of the Rains family.  Rains introduced a PPT-relaxed relative-entropy bound on distillable entanglement~\cite{Rains1999,Rains2001}. Denote the Rains set by
\begin{align}
\mathsf{PPT}'(A{:}B)
\coloneqq
\bigl\{\sigma_{AB}\ge0:\|\sigma_{AB}^\Gamma\|_1\le1\bigr\}.
\label{eq:Rains-set}
\end{align}
Minimizing the Umegaki relative entropy over this set gives the original
Rains relative entropy.  Pairing the same set with the support, or order-zero, divergence gives
\begin{align}
D_0(\rho\|\sigma) \coloneqq-\log\tr(P\sigma),\quad R_{\min}(\rho)\coloneqq\min_{\sigma\in\mathsf{PPT}'(A{:}B)}
D_0(\rho\|\sigma).
\label{eq:min-Rains-relative-entropy}
\end{align}
Wang and Duan introduced this state quantity as the additive SDP bound
$E_M$~\cite{Wang2016c}.  B\"auml \emph{et al.} subsequently used the name
\emph{min-Rains relative entropy} and reserved \emph{min-Rains information}
for the corresponding channel optimization~\cite{Baeuml2019}. Using the definition of $D_0$, the min-Rains relative entropy becomes
\begin{align}
R_{\min}(\rho)
=-\log M(P),
\label{eq:min-rains-definition}
\end{align}
where
\begin{align}
M(P)
&=\max\bigl\{\tr(PV):V\ge0,\ \|V^\Gamma\|_1\le1\bigr\},
\label{eq:min-rains-primal}\\
&\eqt{(i)}\min\bigl\{\|R^\Gamma\|_\infty:R\ge P\bigr\},
\label{eq:min-rains-dual}
\end{align}
where (i) is the strong duality established in Ref.~\cite{Wang2016c} follows from the standard Slater condition~\cite{VandenbergheBoyd1996}. The relation to exact distillation is immediate: every feasible point of Eq.~\eqref{eq:W0-program} is also feasible for Eq.~\eqref{eq:min-rains-dual}.  A second essential property is multiplicativity~\cite{Wang2016c}.  Explicitly, we have $M(P)\le W_0(P)$ and $M(P\ox Q)=M(P)M(Q)$. Combining these two properties gives, for every $n\ge1$,
\begin{align}
-\frac1n\log W_0(P^{\ox n})\le-\frac1n\log M(P^{\ox n})=-\log M(P)
=R_{\min}(\rho).
\label{eq:general-min-rains-bound-block}
\end{align}
Taking the limit yields
\begin{align}
E_{0,\mathrm{PPT}}^\infty(\rho)\le R_{\min}(\rho).
\label{eq:general-min-rains-bound}
\end{align}
The asymptotic tightness question is whether this bound is \textit{always} an equality, i.e., $E_{0,\mathrm{PPT}}^\infty(\rho)\stackrel{?}{=}R_{\min}(\rho)$. It is nontrivial to answer this question since we need to consider the regularization $\lim_{n\to \infty}-1/n\log W_0(P^{\ox n})$ and many of the common candidate states actually give an equality; see the example of antisymmetric states in Appendix~\ref{app:antisym}.

\subsection{Range-supported tensor witness}
Our first main result is to place an upper bound on $E^{(1)}_{0,\mathrm{PPT}}(\rho^{\ox n})$ for every blocklength $n$, based on the fact that the upper constraint $E\le\idop$ fixes every feasible effect on the support, as the following lemma shows.

\begin{lemma}\label{lem:support-rigidity}
Let $P$ be an orthogonal projection and let $P\le E\le\idop$.
Relative to $\cH=\im(P)\oplus\ker(P)$,
\begin{align}
    E=\idop_{\im(P)}\oplus E_\perp,
    \qquad
    0\le E_\perp\le\idop_{\ker(P)}.
    \label{eq:effect-block-form}
\end{align}
Equivalently, $EP=PE=P$.
\end{lemma}

\begin{proof}
For $\ket v\in\im(P)$, the inequalities $P\le E\le\idop$ imply
\begin{align}
    \|v\|^2
    =\langle v|P|v\rangle
    \le\langle v|E|v\rangle
    \le\|v\|^2.
\end{align}
Therefore
$\langle v|(\idop-E)|v\rangle
=\|(\idop-E)^{1/2}\ket v\|^2=0$, and hence
$E\ket v=\ket v$.  Thus $EP=P$.  Taking adjoints gives $PE=P$
because $E$ and $P$ are Hermitian.  These identities fix the first
block column and row of $E$ relative to
$\im(P)\oplus\ker(P)$ and force both off-diagonal blocks to vanish.
Restricting $E\ge0$ and $\idop-E\ge0$ to $\ker(P)$ gives
$0\le E_\perp\le\idop_{\ker(P)}$.
\end{proof}

We now convert this fixed-action property into a bound valid at every blocklength. We call $T$ a range-supported tensor witness because $PT=T$ and it yields a lower bound on $\|E^{\Gamma}\|_\infty$ in Eq.~\eqref{eq:W0-program} that tensorizes under arbitrary tensor powers.

\begin{theorem}[Range-supported tensor witness]
\label{thm:range-supported-witness}
Let $P$ be a nonzero bipartite projection, and let $T$ be an arbitrary
linear operator satisfying
\begin{align}
    PT=T,\qquad \tr T=1.
    \label{eq:range-witness-hypotheses}
\end{align}
Then, for every $n\ge1$,
\begin{align}
    W_0(P^{\ox n})\ge\|T^\Gamma\|_1^{-n}.
    \label{eq:range-witness-block-bound}
\end{align}
Consequently, every state $\rho$ with support projection $P$ satisfies
\begin{align}
    E_{0,\mathrm{PPT}}^\infty(\rho)
    \le\log\|T^\Gamma\|_1.
    \label{eq:rate-range-witness}
\end{align}
\end{theorem}

\begin{proof}
Set $P_n=P^{\ox n}$ and $T_n=T^{\ox n}$.  On $n$ copies, let $\Gamma_n$ denote partial transpose on the joint system $B_1\cdots B_n$; below we again write $\Gamma$ when the blocklength is
clear.  The condition $PT=T$ and $\tr T=1$ imply
\begin{align}
P_nT_n=T_n,\qquad
\tr T_n=(\tr T)^n=1.
\label{eq:tensor-witness-identities}
\end{align}
If $E$ is feasible for $W_0(P_n)$, by~\Cref{lem:support-rigidity}, we have that
$EP_n=P_n$.  Therefore
\begin{align}
ET_n=EP_nT_n=T_n.
\label{eq:tensor-effect-fixes-witness}
\end{align}
Consequently, H\"older's inequality gives
\begin{align}
    1=|\tr T_n|=|\tr(ET_n)|
    =|\tr((E^\Gamma)^\dagger T_n^\Gamma)|
    \le\|E^\Gamma\|_\infty\|T_n^\Gamma\|_1.
    \label{eq:witness-holder-chain}
\end{align}
Finally, by the fact that $\|T_n^\Gamma\|_1=\|(T^\Gamma)^{\ox n}\|_1=\|T^\Gamma\|_1^n$, we have that $\|E^\Gamma\|_\infty\ge\|T^\Gamma\|_1^{-n}$, proving Eq.~\eqref{eq:range-witness-block-bound}.  Taking $-n^{-1}\log$ and then the limit gives Eq.~\eqref{eq:rate-range-witness}.
\end{proof}

Pairing the range witness with a feasible min-Rains dual certificate gives a direct criterion for strict separation.
\begin{corollary}
\label{cor:two-certificate}
Let $P$ be a nonzero bipartite projection.  Suppose that operators $T$
and $R$ satisfy
\begin{align}
    PT=T,\qquad
    \tr T=1,\qquad
    R\ge P,\qquad
    \|T^\Gamma\|_1\|R^\Gamma\|_\infty<1.
    \label{eq:two-certificate-hypotheses}
\end{align}
Then every state $\rho$ with support projection $P$ satisfies
\begin{align}
    E_{0,\mathrm{PPT}}^\infty(\rho)<R_{\min}(\rho).
    \label{eq:entropic-separation-general}
\end{align}
\end{corollary}

\begin{proof}
Theorem~\ref{thm:range-supported-witness} gives
$E_{0,\mathrm{PPT}}^\infty(\rho)\le\log\|T^\Gamma\|_1$.
The condition $R\ge P$ makes $R$ feasible in
\eqref{eq:min-rains-dual}; hence
\begin{align}
    M(P)\le\|R^\Gamma\|_\infty,
    \qquad
    -\log\|R^\Gamma\|_\infty\le R_{\min}(\rho).
\end{align}
The strict product condition is equivalent to
$\log\|T^\Gamma\|_1<-\log\|R^\Gamma\|_\infty$.  Combining the
three bounds gives
\begin{align}
    E_{0,\mathrm{PPT}}^\infty(\rho)
    \le\log\|T^\Gamma\|_1
    <-\log\|R^\Gamma\|_\infty
    \le R_{\min}(\rho).
\end{align}
\end{proof}

\section{An Explicit Qutrit--Qutrit Counterexample}
\label{sec:explicit-example}

The support is built from the three undirected edges of a qutrit cycle.  Each
edge contains the two oppositely oriented computational basis vectors.  The
three edge subspaces are mutually orthogonal, while partial transpose moves all
edge coherences into the three-dimensional diagonal subspace.  This produces a
small nontrivial block in which phases and weights can be controlled exactly.

In the computational basis, define
\begin{align}
u_0&=\frac{8\ket{01}+7\ket{10}}{\sqrt{113}},&
u_1&=\frac{7\ket{12}+8\ket{21}}{\sqrt{113}},
\label{eq:counterexample-vectors-first}\\
u_2&=\frac{\ket{20}+i\ket{02}}{\sqrt2},&
P&=\sum_{j=0}^2\ket{u_j}\!\bra{u_j},\qquad
\rho_\star=\frac13P.
\label{eq:counterexample-state}
\end{align}
The vectors $u_0,u_1,u_2$ are normalized and occupy mutually orthogonal
edge subspaces.  Hence they are orthonormal,
$P^2=P=P^\dagger$, and $\tr P=3$.  It follows that
$\rho_\star=P/3\ge0$ and $\tr\rho_\star=1$, so $\rho_\star$ is one
rank-three state with support projection $P$.  For a subspace $\cH$ with ordered basis
$\mathcal B=(\ket{b_1},\ldots,\ket{b_d})$, use convention
\begin{align}
\left[A\big|_{\cH}\right]_{\mathcal B}=(a_{jk})
\quad\Longleftrightarrow\quad
A\ket{b_k}=\sum_j a_{jk}\ket{b_j}.
\label{eq:restriction-matrix-convention}
\end{align}
Introduce the mutually orthogonal subspaces
\begin{align}
\cH_{\mathrm{diag}}
&\coloneqq\operatorname{span}\{\ket{00},\ket{11},\ket{22}\},
&\mathcal B_{\mathrm{diag}}&\coloneqq(\ket{00},\ket{11},\ket{22}),
\notag\\
\cH_{01}&\coloneqq\operatorname{span}\{\ket{01},\ket{10}\},
&\mathcal B_{01}&\coloneqq(\ket{01},\ket{10}),
\notag\\
\cH_{12}&\coloneqq\operatorname{span}\{\ket{12},\ket{21}\},
&\mathcal B_{12}&\coloneqq(\ket{12},\ket{21}),
\notag\\
\cH_{02}&\coloneqq\operatorname{span}\{\ket{02},\ket{20}\},
&\mathcal B_{02}&\coloneqq(\ket{02},\ket{20}).
\label{eq:block-bases}
\end{align}
Thus
\begin{align}
\CC^3\ox\CC^3
=\cH_{\mathrm{diag}}\oplus\cH_{01}\oplus\cH_{12}\oplus\cH_{02}.
\label{eq:block-decomposition}
\end{align}
In this decomposition, $P$ vanishes on $\cH_{\mathrm{diag}}$ and has the three
blocks
\begin{align}
\left[P\big|_{\cH_{01}}\right]_{\mathcal B_{01}}
&=\frac1{113}\begin{pmatrix}64&56\\56&49\end{pmatrix},
\label{eq:P-block-01}\\
\left[P\big|_{\cH_{12}}\right]_{\mathcal B_{12}}
&=\frac1{113}\begin{pmatrix}49&56\\56&64\end{pmatrix},
\label{eq:P-block-12}\\
\left[P\big|_{\cH_{02}}\right]_{\mathcal B_{02}}
&=\frac12\begin{pmatrix}1&i\\-i&1\end{pmatrix}.
\label{eq:P-block-02}
\end{align}
We display them because all later semidefinite comparisons reduce to these
three $2\times2$ blocks.

\begin{theorem}[Explicit qutrit--qutrit separation]
\label{thm:qutrit-counterexample}
Let $P$ be the rank-three projection in
\eqref{eq:counterexample-state}.  Every bipartite state $\rho$ with
support projection $P$ satisfies
\begin{align}
    E_{0,\mathrm{PPT}}^\infty(\rho)
    \le\log\|T^\Gamma\|_1
    <\log\frac{391}{250}
    <-\log\frac{6393}{10000}
    \le R_{\min}(\rho).
    \label{eq:strict-counterexample-chain}
\end{align}
Here $T$ is the range witness defined in Eqs.~\eqref{eq:xyz-definitions}--\eqref{eq:T-block-02}.
In particular,
$E_{0,\mathrm{PPT}}^\infty(\rho)<R_{\min}(\rho)$.
\end{theorem}

The following subsections verify the two certificates exactly; the numerical
search that suggested their sparse form is reconstructed in
Section~\ref{sec:discovery}.
Fix an arbitrary state $\rho$ with support projection $P$ for the remainder
of the proof.

\subsection{Certificate for the upper bound}
We first construct a one-sided range witness $T$, verify
$PT=T$ and $\tr T=1$, and then compute the singular values of
$T^\Gamma$ to prove $\|T^\Gamma\|_1<391/250$.  Set
\begin{align}
x&=\frac{1319}{5000},\qquad y=\frac{444}{5000},\\
z&=\frac{737}{5000},\qquad x+y+z=\frac12.
\label{eq:xyz-definitions}
\end{align}
Define $T$ to vanish on $\cH_{\mathrm{diag}}$.  Its other blocks are
\begin{align}
\left[T\big|_{\cH_{01}}\right]_{\mathcal B_{01}}
&=\begin{pmatrix}x&8y/7\\7x/8&y\end{pmatrix},
\label{eq:T-block-01}\\
\left[T\big|_{\cH_{12}}\right]_{\mathcal B_{12}}
&=\begin{pmatrix}y&7x/8\\8y/7&x\end{pmatrix},
\label{eq:T-block-12}\\
\left[T\big|_{\cH_{02}}\right]_{\mathcal B_{02}}
&=\begin{pmatrix}z&iz\\-iz&z\end{pmatrix}.
\label{eq:T-block-02}
\end{align}
Each block has rank one, and every column belongs to the corresponding
one-dimensional component of $\im(P)$.  This can be checked by
\begin{align}
\begin{pmatrix}x&8y/7\\7x/8&y\end{pmatrix}
&=\begin{pmatrix}8\\7\end{pmatrix}
\begin{pmatrix}x/8&y/7\end{pmatrix},
\label{eq:T-factor-01}\\
\begin{pmatrix}y&7x/8\\8y/7&x\end{pmatrix}
&=\begin{pmatrix}7\\8\end{pmatrix}
\begin{pmatrix}y/7&x/8\end{pmatrix},
\label{eq:T-factor-12}\\
\begin{pmatrix}z&iz\\-iz&z\end{pmatrix}
&=\begin{pmatrix}i\\1\end{pmatrix}
\begin{pmatrix}-iz&z\end{pmatrix}.
\label{eq:T-factor-02}
\end{align}
The left factors in Eqs.~\eqref{eq:T-factor-01}--\eqref{eq:T-factor-02} are
proportional to the coordinate vectors of $u_0,u_1,u_2$ in their edge bases.
The range of a rank-one matrix $ab^{\mathsf T}$ is contained in
$\operatorname{span}\{a\}$.  Hence the range of each block of $T$ lies in the
corresponding support line, and the vanishing diagonal block contributes no
additional range.  Therefore $\im(T)\subseteq\im(P)$.  Since $P$ acts as
the identity on every column of $T$,
\begin{align}
PT=T.
\label{eq:PT-equals-T-explicit}
\end{align}
Moreover,
\begin{align}
\tr T=2(x+y+z)=1.
\label{eq:trace-T-explicit}
\end{align}
It remains to calculate $\|T^\Gamma\|_1$.  Put
\begin{align}
\alpha=\frac{8y}{7},\qquad \beta=\frac{7x}{8}.
\end{align}
To calculate the partial transpose, apply
$\ketbra{ab}{cd}\mapsto\ketbra{ad}{cb}$ to every matrix unit.  The terms $\ketbra{ab}{ab}$ are unchanged.  Therefore, on the six ordered vectors
$\ket{01},\ket{10},\ket{12},\ket{21},\ket{02},\ket{20}$,
$T^\Gamma$ acts by the positive scalars $x,y,y,x,z,z$.  Their total
contribution to the trace norm is $2(x+y+z)=1$.  Each coherence
$\ket{ab}\!\bra{ba}$ is instead mapped to
$\ket{aa}\!\bra{bb}$, so all six coherence terms assemble on
$\cH_{\mathrm{diag}}$.  Reading their coefficients in the ordered basis
$\mathcal B_{\mathrm{diag}}$ gives
\begin{align}
K\coloneqq\left[T^\Gamma\big|_{\cH_{\mathrm{diag}}}\right]_{\mathcal B_{\mathrm{diag}}}
=\begin{pmatrix}
    0&\alpha&iz\\
    \beta&0&\beta\\
    -iz&\alpha&0
\end{pmatrix}.
\label{eq:K-block}
\end{align}
Thus $T^\Gamma=\operatorname{diag}(x,y,y,x,z,z)\oplus K$.  Since the trace norm is additive on orthogonal direct sums, it remains to calculate
$\|K\|_1$.  Notice that
\begin{align}
K={}&
\begin{pmatrix}iz/2\\ \beta\\-iz/2\end{pmatrix}
\begin{pmatrix}1&0&1\end{pmatrix}
+
\begin{pmatrix}1\\0\\1\end{pmatrix}
\begin{pmatrix}-iz/2&\alpha&iz/2\end{pmatrix}
=a_1b_1^\dagger+a_2b_2^\dagger,
\label{eq:K-rank-one-decomposition}
\end{align}
where
\begin{align}
a_1&=(iz/2,\ \beta,\ -iz/2)^{\mathsf T},&
b_1&=(1,\ 0,\ 1)^{\mathsf T},\notag\\
a_2&=(1,\ 0,\ 1)^{\mathsf T},&
b_2&=(iz/2,\ \alpha,\ -iz/2)^{\mathsf T}.
\label{eq:K-orthogonal-factors}
\end{align}
Direct calculation gives $a_1^\dagger a_2=0$ and
$b_1^\dagger b_2=0$.  Thus the two rank-one maps have orthogonal output spaces
and orthogonal input spaces.  In orthonormal bases adapted to these spaces,
$K$ is the direct sum of two rank-one scalar blocks and one zero block.  Its
two nonzero singular values are therefore
\begin{align}
\|a_1\|\|b_1\|
=\sqrt{\left(\beta^2+\frac{z^2}{2}\right)2}
=\sqrt{2\beta^2+z^2},\qquad
\|a_2\|\|b_2\|
=\sqrt{2\left(\alpha^2+\frac{z^2}{2}\right)}
=\sqrt{2\alpha^2+z^2}.
\end{align}
Hence the singular-value multiset of $K$ is
\begin{align}
\{0,\sigma_1,\sigma_2\},
\qquad
\sigma_1=\sqrt{2\beta^2+z^2},\quad
\sigma_2=\sqrt{2\alpha^2+z^2}.
\label{eq:K-singular-values}
\end{align}
Equivalently,
\begin{align}
\sigma_1^2=\frac{49}{32}x^2+z^2,\qquad
\sigma_2^2=\frac{128}{49}y^2+z^2.
\end{align}
The exact comparisons
\begin{align}
\left(\frac{1791}{5000}\right)^2-\sigma_1^2
=\frac{3219}{160000000}>0,
\qquad
\left(\frac{1029}{5000}\right)^2-\sigma_2^2
=\frac{863}{30625000}>0
\label{eq:sigma2-margin}
\end{align}
therefore imply
$\sigma_1<1791/5000$ and $\sigma_2<1029/5000$, since all quantities in Eq.~\eqref{eq:K-singular-values} are nonnegative. Consequently,
\begin{align}
\mu\coloneqq\|T^\Gamma\|_1
=1+\sigma_1+\sigma_2
<1+\frac{1791+1029}{5000}
=\frac{391}{250}.
\label{eq:T-pt-norm}
\end{align}
Theorem~\ref{thm:range-supported-witness} and the strict bound on $\mu$ give
\begin{align}
E_{0,\mathrm{PPT}}^\infty(\rho)
\le\log\mu<\log\frac{391}{250}.
\label{eq:rate-upper-counterexample}
\end{align}

\subsection{Certificate for the min-Rains lower bound}
Second, we construct a Hermitian operator $R$, prove
$R-P\ge0$ block by block, and compute
$\|R^\Gamma\|_\infty=6393/10000$.
Set
\begin{align}
m=\frac{6393}{10000},\qquad
r=\frac{933}{2500},\qquad
e=\frac{3607}{10000}.
\label{eq:mre-definitions}
\end{align}
Define a Hermitian operator $R$ that vanishes on $\cH_{\mathrm{diag}}$ and has
the edge blocks
\begin{align}
\left[R\big|_{\cH_{01}}\right]_{\mathcal B_{01}}
=\left[R\big|_{\cH_{12}}\right]_{\mathcal B_{12}}
&=\begin{pmatrix}m&r\\r&m\end{pmatrix},
\label{eq:R-real-blocks}\\
\left[R\big|_{\cH_{02}}\right]_{\mathcal B_{02}}
&=\begin{pmatrix}m&ie\\-ie&m\end{pmatrix}.
\label{eq:R-complex-block}
\end{align}
The two real blocks are symmetric, and the off-diagonal entries $ie$ and $-ie$
in the third block are complex conjugates.  Thus every block is Hermitian and
so is their orthogonal direct sum $R$.
This certificate is intentionally unavailable to the exact-distillation
program: each real block has eigenvalue $m+r=\frac{81}{80}>1$, so $R\nleq\idop$. We next verify the required domination
$R\ge P$.

The difference $R-P$ is block diagonal with respect to
Eq.~\eqref{eq:block-decomposition}, because both $R$ and $P$ preserve every
summand.  It vanishes on $\cH_{\mathrm{diag}}$, where both operators vanish.
On each of $\cH_{01}$ and $\cH_{12}$, its two diagonal entries are
\begin{align}
m-\frac{64}{113}
&=\frac{82409}{1130000}>0,
\label{eq:RminusP-diagonal1}\\
m-\frac{49}{113}
&=\frac{232409}{1130000}>0.
\label{eq:RminusP-diagonal2}
\end{align}
The off-diagonal entry on either real block is $r-56/113$.  Hence their common
determinant is
\begin{equation}\label{eq:RminusP-real-determinant}
\begin{aligned}
    \left(m-\frac{64}{113}\right)
    \left(m-\frac{49}{113}\right)
    -\left(r-\frac{56}{113}\right)^2
    =\frac{2133}{90400000}>0.
\end{aligned}
\end{equation}
A $2\times2$ Hermitian matrix
$\left(\begin{smallmatrix}a&c\\\overline c&d\end{smallmatrix}\right)$ is
positive definite when $a>0$ and $ad-|c|^2>0$. The second condition then forces $d>0$.  Eq.~\eqref{eq:RminusP-diagonal1} and
Eq.~\eqref{eq:RminusP-real-determinant} therefore prove that both real blocks of
$R-P$ are positive definite.  On $\cH_{02}$, subtracting
Eq.~\eqref{eq:P-block-02} from Eq.~\eqref{eq:R-complex-block} gives
\begin{align}
\left[(R-P)\big|_{\cH_{02}}\right]_{\mathcal B_{02}}
=\begin{pmatrix}h&-ih\\ih&h\end{pmatrix},
\quad
h=m-\frac12=\frac{1393}{10000}>0,
\label{eq:RminusP-complex-block}
\end{align}
where $e-1/2=-h$.  The displayed matrix equals
$h\ket v\!\bra v$ for the coordinate vector
$v=(1,i)^{\mathsf T}$, whose squared norm is two.  Its eigenvalues are
therefore $0$ and $2h$, both nonnegative.  All four mutually orthogonal blocks
of $R-P$ are now positive semidefinite.  A block-diagonal Hermitian operator is
positive semidefinite if and only if every diagonal block is positive
semidefinite, so
\begin{align}
R-P\ge0.
\label{eq:R-dominates-P}
\end{align}

After partial transpose, the six population entries of the three edge blocks
remain on their original computational basis vectors and all equal $m$.
The coherence entries move to $\cH_{\mathrm{diag}}$ by the rule used above for
$T$.  Consequently, $R^\Gamma$ equals $m\idop$ on the six-dimensional
direct sum of the edge coordinates, and its remaining block is
\begin{align}
G\coloneqq\left[R^\Gamma\big|_{\cH_{\mathrm{diag}}}\right]_{\mathcal B_{\mathrm{diag}}}
=\begin{pmatrix}
    0&r&ie\\
    r&0&r\\
    -ie&r&0
\end{pmatrix}.
\label{eq:G-block}
\end{align}
The matrix $G$ is Hermitian.  Expanding its characteristic polynomial gives
\begin{align}
\det(\lambda\idop-G)
=\lambda\bigl[\lambda^2-(2r^2+e^2)\bigr].
\label{eq:G-characteristic-polynomial}
\end{align}
Its three eigenvalues are therefore
\begin{align}
\operatorname{spec}(G)
=\{0,\ \pm\sqrt{2r^2+e^2}\}.
\label{eq:G-spectrum}
\end{align}
The exact margin
\begin{align}
m^2-(2r^2+e^2)=\frac{17}{390625}>0
\label{eq:R-pt-margin}
\end{align}
shows that $m^2>2r^2+e^2$.  Since $m>0$, we have $m>\sqrt{2r^2+e^2}$. Thus the spectral radius of $G$ is strictly smaller than $m$. The other six eigenvalues of the Hermitian operator $R^\Gamma$ equal $m$. Hence
\begin{align}
\|R^\Gamma\|_\infty=m=\frac{6393}{10000}.
\label{eq:R-pt-norm}
\end{align}
Eq.~\eqref{eq:R-dominates-P} makes $R$ feasible for the minimization in Eq.~\eqref{eq:min-rains-dual}. It follows that
\begin{align}
M(P)\le\frac{6393}{10000}.
\label{eq:M-upper-counterexample}
\end{align}

\begin{proof}[Proof of Theorem~\ref{thm:qutrit-counterexample}]
Eq.~\eqref{eq:PT-equals-T-explicit} and Eq.~\eqref{eq:trace-T-explicit} show that the range certificate satisfies exactly
the hypotheses of Theorem~\ref{thm:range-supported-witness}.  Its norm estimate
\eqref{eq:T-pt-norm} therefore proves
\eqref{eq:rate-upper-counterexample}.  Independently,
\eqref{eq:R-dominates-P} and~\eqref{eq:R-pt-norm} make $R$ a feasible
min-Rains dual point and prove~\eqref{eq:M-upper-counterexample}.

It remains only to verify that the two certified bounds are strictly ordered.
Putting them over the common denominator $3910000$ gives
\begin{align}
    \frac{250}{391}-\frac{6393}{10000}
    =\frac{2500000-2499663}{3910000}
    =\frac{337}{3910000}>0.
    \label{eq:explicit-rational-gap-expanded}
\end{align}
Combining the range-witness rate bound, the estimate on $\mu$, the exact
rational comparison, and the dual-certificate bound gives
\begin{align}
    E_{0,\mathrm{PPT}}^\infty(\rho)
    \leqt{(i)}\log\mu
    \lt{(ii)}\log\frac{391}{250}
    \lt{(iii)}-\log\frac{6393}{10000}
    \eqt{(iv)}-\log\|R^\Gamma\|_\infty
    \le R_{\min}(\rho).
    \label{eq:decisive-exact-chain}
\end{align}
Here (i) follows from Theorem~\ref{thm:range-supported-witness};
(ii) is Eq.~\eqref{eq:T-pt-norm}; (iii) is
Eq.~\eqref{eq:explicit-rational-gap-expanded}; (iv) is
Eq.~\eqref{eq:R-pt-norm}; and the last inequality follows from dual
feasibility of $R$.
\end{proof}

\section{How the Counterexample Was Found}
\label{sec:discovery}

We now reconstruct the discovery procedure in enough detail to explain both the support and the
apparently irregular rational coefficients.  Numerics were used in three distinct roles:
\begin{itemize}
    \item To search over a structured family of state supports $P$.
    \item To identify $T$ on a fixed support.
    \item To locate points with sufficient numerical margin for rational approximation.
\end{itemize}

\subsection{A nested search over weighted cycle supports}

We restricted the outer search to the weighted triangle-cycle family
\begin{align}
u_j(\theta_j,\phi_j)
=\cos\theta_j\ket{j,j+1}
+e^{i\phi_j}\sin\theta_j\ket{j+1,j},
\quad j\in\mathbb Z_3,
\label{eq:weighted-cycle-family}
\end{align}
where both indices are understood modulo three.  The associated support
projection is
$P(\boldsymbol\theta,\boldsymbol\phi)
=\sum_j\ket{u_j}\!\bra{u_j}$.  The three vectors occupy disjoint pairs of
computational basis states and are therefore orthonormal for every choice of
parameters.  The family is also adapted to partial transpose.  Populations
remain on the six oriented edge vectors, whereas every edge coherence moves
to the diagonal qutrit subspace; for example,
$(\ket{01}\!\bra{10})^\Gamma=\ket{00}\!\bra{11}$.  Consequently, every
edge-respecting operator becomes six scalar blocks plus one $3\times3$
coherence block after partial transpose.  The weights control the domination constraints on the edge blocks, while the phases control interference in the coherence block.

The witness $T$ need not be Hermitian.  We therefore use the standard semidefinite representation of the trace norm, which is valid for every rectangular complex matrix (see, e.g.,~\cite[Example~1.20]{Watrous2018})
\begin{align}
\|A\|_1
=\min_{X,Y\ge0}\left\{
\frac12\bigl(\tr X+\tr Y\bigr):
\begin{pmatrix}
    X&A\\
    A^\dagger&Y
\end{pmatrix}\ge0
\right\}.
\label{eq:general-trace-norm-sdp}
\end{align}
For each support in the outer family, we evaluated two inner convex programs.
The first was the one-sided range-witness cost
\begin{align}
\mu_\star(P)
\coloneqq\min\bigl\{\|T^\Gamma\|_1:PT=T,\ \tr T=1\bigr\}.
\label{eq:discovery-witness-program}
\end{align}
Applying Eq.~\eqref{eq:general-trace-norm-sdp} to $A=T^\Gamma$ gives the
explicit SDP
\begin{align}
\mu_\star(P)=\min_{T,Y,Z}\quad
&\frac12\tr(Y+Z)\notag\\
\text{subject to}\quad
&PT=T,\qquad \tr T=1,\notag\\
&Y,Z\ge0,\qquad
\begin{pmatrix}
    Y&T^\Gamma\\
    (T^\Gamma)^\dagger&Z
\end{pmatrix}\ge0.
\label{eq:discovery-witness-sdp}
\end{align}
The second inner problem was the min-Rains support program $M(P)$ in
Eqs.~\eqref{eq:min-rains-primal} and~\eqref{eq:min-rains-dual}.  We used
\begin{align}
\mathcal S(P)\coloneqq\mu_\star(P)M(P)
\label{eq:discovery-support-score}
\end{align}
as the search score.  The computational loop formed $P$ from one set of six outer parameters, solved the two inner SDPs, and returned $\mathcal S(P)$ to the outer search.  A value below one predicts that the two programs admit strictly separated certificates.  Optimization of $\mathcal S$ over $(\boldsymbol\theta,\boldsymbol\phi)$ is nonconvex and was used only to find candidates.  The search produced a candidate close to the angles and phases
\begin{align}
(\theta_0,\theta_1,\theta_2)=\left(\arctan\frac78,\arctan\frac87,\frac\pi4\right),\quad (\phi_0,\phi_1,\phi_2)&=\left(0,0,\frac\pi2\right).
\label{eq:exact-cycle-parameters}
\end{align}
The first two edges therefore have the reciprocal weight ratios $8:7$ and
$7:8$, while the last edge has equal weights and a quarter-turn phase.  The
latter phase makes the product of the three coherence phases purely
imaginary.  This is why the nontrivial blocks in Eq.~\eqref{eq:K-block} and Eq.~\eqref{eq:G-block} respectively acquire an orthogonal rank-one decomposition and the simple spectrum
$\{0,\pm\sqrt{2r^2+e^2}\}$.  Replacing the numerical angles by Eq.~\eqref{eq:exact-cycle-parameters} gives exactly the vectors in Eq.~\eqref{eq:counterexample-vectors-first}.

\subsection{Reduced search for the one-sided witness}

After fixing the support in~\eqref{eq:counterexample-state}, the numerical
solution of~\eqref{eq:discovery-witness-sdp} suggested one rank-one block on
each edge.  The range condition largely forces this block pattern.  For
example, the first support line has coordinate vector $(8,7)^{\mathsf T}$.
A rank-one block with this output range and diagonal entries $x,y$ must be
\begin{align}
\begin{pmatrix}8\\7\end{pmatrix}
\begin{pmatrix}x/8&y/7\end{pmatrix}
=
\begin{pmatrix}x&8y/7\\7x/8&y\end{pmatrix}.
\label{eq:discovery-forced-T-block}
\end{align}
The other two edge blocks follow in the same way.  Within the positive three-parameter block pattern of Eqs.~\eqref{eq:T-block-01}--\eqref{eq:T-block-02}, trace normalization gives $x+y+z=1/2$, and the singular-value calculation in Eq.~\eqref{eq:K-singular-values} reduces the numerical witness search to
\begin{align}
\min_{\substack{x,y,z\ge0\\x+y+z=1/2}}
\left[
1+\sqrt{\frac{49}{32}x^2+z^2}
+\sqrt{\frac{128}{49}y^2+z^2}
\right].
\label{eq:reduced-witness-search}
\end{align}
A numerical minimizer of this restricted problem is
\begin{align}
(x,y,z)\approx
(0.263786245,\ 0.088832207,\ 0.147381548),
\label{eq:numerical-witness-candidate}
\end{align}
with objective value approximately $1.563903435$.  The rational point in Eq.~\eqref{eq:xyz-definitions} was selected nearby.  Its numerators sum to $2500$, so $x+y+z=1/2$ remains exact, and its restricted objective differs from the numerical value in Eq.~\eqref{eq:numerical-witness-candidate} by less
than $7\times10^{-9}$.  The fractions $1791/5000$ and $1029/5000$ were then chosen as strict rational upper bounds on the two singular values.

\subsection{Reduced search for the min-Rains certificate}

The numerical min-Rains solution suggested the three-parameter Hermitian ansatz in Eqs.~\eqref{eq:R-real-blocks}--\eqref{eq:R-complex-block}.  For this ansatz, the condition $R-P\ge0$ is equivalent to the edge constraints
\begin{align}
m&\ge\frac{64}{113},\qquad m\ge\frac{49}{113},
\label{eq:reduced-R-diagonal-constraints}\\
\left(m-\frac{64}{113}\right)
\left(m-\frac{49}{113}\right)
&\ge\left(r-\frac{56}{113}\right)^2,
\label{eq:reduced-R-real-constraint}\\
m-\frac12&\ge\left|e-\frac12\right|.
\label{eq:reduced-R-complex-constraint}
\end{align}
Moreover, the partial-transpose calculation gives six eigenvalues equal to
$m$ and three eigenvalues $0,\pm\sqrt{2r^2+e^2}$.  Hence
$\|R^\Gamma\|_\infty=\max\{m,\sqrt{2r^2+e^2}\}$.  Imposing the spectral
inequality below makes the objective equal to $m$, so the restricted search
can be written as
\begin{align}
\min_{m,r,e\ge0}\quad&m\notag\\
\text{subject to}\quad
&\eqref{eq:reduced-R-diagonal-constraints}
\text{--}\eqref{eq:reduced-R-complex-constraint},\notag\\
&2r^2+e^2\le m^2.
\label{eq:reduced-R-search}
\end{align}
The geometry of these constraints explains the final numbers.  On the lower
branches relevant to the numerical candidate, fixed $m$ gives $e=1-m$ when
\eqref{eq:reduced-R-complex-constraint} is active and gives
\begin{align}
r=\frac{56}{113}
-\sqrt{\left(m-\frac{64}{113}\right)
    \left(m-\frac{49}{113}\right)}.
\label{eq:reduced-R-boundary-r}
\end{align}
The numerical candidate at the lower boundary also saturates the spectral
constraint.  Solving these three boundary relations numerically gives
\begin{align}
(m,r,e)\approx
(0.639249331,\ 0.373161267,\ 0.360750669).
\label{eq:numerical-R-candidate}
\end{align}
The rational values in Eq.~\eqref{eq:mre-definitions} move this numerical boundary point to the feasible side of both the determinant and spectral constraints.

\section{Conclusion}
\label{sec:discussion}

We have shown that the additive min-Rains relative entropy is not an exact formula for regularized exact PPT distillable entanglement.  The separation is caused by a tensor-stable geometric constraint: an exact-distillation effect satisfying $P\le E\le\idop$ is fixed to the identity on the input support, whereas a min-Rains dual operator needs only to dominate $P$ and may exceed the identity.  A one-sided range witness converts this difference into a multiplicative bound, and an explicit rank-three qutrit support admits exact certificates with strictly ordered values. It still remains an interesting open problem to determine what the closed-form formula for the exact PPT distillable entanglement is if it is not the min-Rains relative entropy.

\section*{Acknowledgements}
The central conceptual ideas, structural inequalities, and proof strategy were developed by the authors. Large language models were used to turn~\Cref{cor:two-certificate} into the concrete counterexample, including carrying out the numerical optimization, identifying candidate coefficients, and assisting with their exact rational certification.  The authors reviewed and verified the AI-assisted identities, inequalities, and certificates and take full responsibility for the results. We thank
Kun Fang, Ludovico Lami, and Bartosz Regula for many discussions on
this problem. This work was partially supported by the National Natural Science Foundation of China (Grant Nos.~92576114, 12447107) and the Guangdong Provincial Quantum Science Strategic Initiative (Grant Nos.~GDZX2403008, GDZX2503001, and GDZX2403001).

\bibliographystyle{alpha}
\bibliography{ref}

@book{Watrous2018,
  author    = {John Watrous},
  title     = {The Theory of Quantum Information},
  publisher = {Cambridge University Press},
  year      = {2018},
  doi       = {10.1017/9781316848142}
}

@book{Hayashi2017b,
address = {Berlin, Heidelberg},
author = {Hayashi, Masahito},
doi = {10.1007/978-3-662-49725-8},
isbn = {978-3-662-49723-4},
publisher = {Springer Berlin Heidelberg},
series = {Graduate Texts in Physics},
title = {{Quantum Information Theory}},
url = {http://link.springer.com/10.1007/978-3-662-49725-8},
year = {2017}
}

@article{Khatri2025a,
author = {Khatri, Sumeet and Lami, Ludovico and Wilde, Mark M},
title = {{Principles of Quantum Communication Theory: A Modern Approach}},
year = {2025}
}

@article{XWthesis,
author = {Wang, Xin},
journal = {PhD thesis},
title = {{Semidefinite optimization for quantum information}},
url = {https://opus.lib.uts.edu.au/handle/10453/127996},
year = {2018}
}

@article{Matthews2008,
  title = {Pure-state transformations and catalysis under operations that completely preserve positivity of partial transpose},
  author = {Matthews, William and Winter, Andreas},
  journal = {Phys. Rev. A},
  volume = {78},
  issue = {1},
  pages = {012317},
  numpages = {8},
  year = {2008},
  month = {Jul},
  publisher = {American Physical Society},
  doi = {10.1103/PhysRevA.78.012317},
  url = {https://link.aps.org/doi/10.1103/PhysRevA.78.012317}
}

@article{Duan2005,
  title = {Efficiency of deterministic entanglement transformation},
  author = {Duan, Runyao and Feng, Yuan and Ji, Zhengfeng and Ying, Mingsheng},
  journal = {Phys. Rev. A},
  volume = {71},
  issue = {2},
  pages = {022305},
  numpages = {6},
  year = {2005},
  month = {Feb},
  publisher = {American Physical Society},
  doi = {10.1103/PhysRevA.71.022305},
  url = {https://link.aps.org/doi/10.1103/PhysRevA.71.022305}
}

@article{Rains1999r,
  title = {Rigorous treatment of distillable entanglement},
  author = {Rains, E. M.},
  journal = {Phys. Rev. A},
  volume = {60},
  issue = {1},
  pages = {173--178},
  numpages = {0},
  year = {1999},
  month = {Jul},
  publisher = {American Physical Society},
  doi = {10.1103/PhysRevA.60.173},
  url = {https://link.aps.org/doi/10.1103/PhysRevA.60.173}
}

@article{Wang2016c,
author = {Wang, Xin and Duan, Runyao},
doi = {10.1103/PhysRevA.95.062322},
issn = {2469-9926},
journal = {Physical Review A},
month = {jun},
number = {6},
pages = {062322},
title = {{Nonadditivity of Rains' bound for distillable entanglement}},
url = {http://link.aps.org/doi/10.1103/PhysRevA.95.062322},
volume = {95},
year = {2017}
}

@article{Wang2016,
author = {Wang, Xin and Duan, Runyao},
doi = {10.1103/PhysRevA.94.050301},
journal = {Physical Review A},
month = {nov},
number = {5},
pages = {050301},
title = {{Improved semidefinite programming upper bound on distillable entanglement}},
url = {https://doi.org/10.1103/PhysRevA.94.050301},
volume = {94},
year = {2016}
}

@article{Rains1999,
  title = {Bound on distillable entanglement},
  author = {Rains, E. M.},
  journal = {Phys. Rev. A},
  volume = {60},
  issue = {1},
  pages = {179--184},
  numpages = {0},
  year = {1999},
  month = {Jul},
  publisher = {American Physical Society},
  doi = {10.1103/PhysRevA.60.179},
  url = {https://link.aps.org/doi/10.1103/PhysRevA.60.179}
}

@article{Rains2001,
  author = {Rains, Eric M.},
  title = {A semidefinite program for distillable entanglement},
  journal = {IEEE Transactions on Information Theory},
  volume = {47},
  number = {7},
  pages = {2921--2933},
  year = {2001},
  doi = {10.1109/18.959270}
}

@article{Bennett1996Purification,
  author = {Bennett, Charles H. and Brassard, Gilles and Popescu, Sandu and Schumacher, Benjamin and Smolin, John A. and Wootters, William K.},
  title = {Purification of Noisy Entanglement and Faithful Teleportation via Noisy Channels},
  journal = {Physical Review Letters},
  volume = {76},
  number = {5},
  pages = {722--725},
  year = {1996},
  doi = {10.1103/PhysRevLett.76.722}
}

@article{Bennett1996Mixed,
  author = {Bennett, Charles H. and DiVincenzo, David P. and Smolin, John A. and Wootters, William K.},
  title = {Mixed-State Entanglement and Quantum Error Correction},
  journal = {Physical Review A},
  volume = {54},
  number = {5},
  pages = {3824--3851},
  year = {1996},
  doi = {10.1103/PhysRevA.54.3824}
}

@article{Horodecki2009,
  author = {Horodecki, Ryszard and Horodecki, Pawe{\l} and Horodecki, Micha{\l} and Horodecki, Karol},
  title = {Quantum Entanglement},
  journal = {Reviews of Modern Physics},
  volume = {81},
  number = {2},
  pages = {865--942},
  year = {2009},
  doi = {10.1103/RevModPhys.81.865}
}

@article{Fang2019,
  author = {Fang, Kun and Wang, Xin and Tomamichel, Marco and Duan, Runyao},
  title = {Non-Asymptotic Entanglement Distillation},
  journal = {IEEE Transactions on Information Theory},
  volume = {65},
  number = {10},
  pages = {6454--6465},
  year = {2019},
  doi = {10.1109/TIT.2019.2914688}
}

@misc{Baeuml2019,
  author = {B{\"a}uml, Stefan and Das, Siddhartha and Wang, Xin and Wilde, Mark M.},
  title = {Resource Theory of Entanglement for Bipartite Quantum Channels},
  year = {2019},
  eprint = {1907.04181},
  archivePrefix = {arXiv},
  primaryClass = {quant-ph},
  doi = {10.48550/arXiv.1907.04181}
}

@article{Bennett1993Teleportation,
  author = {Bennett, Charles H. and Brassard, Gilles and Cr{\'e}peau, Claude and Jozsa, Richard and Peres, Asher and Wootters, William K.},
  title = {Teleporting an Unknown Quantum State via Dual Classical and {Einstein--Podolsky--Rosen} Channels},
  journal = {Physical Review Letters},
  volume = {70},
  number = {13},
  pages = {1895--1899},
  year = {1993},
  doi = {10.1103/PhysRevLett.70.1895}
}

@article{Ekert1991,
  author = {Ekert, Artur K.},
  title = {Quantum Cryptography Based on {Bell's} Theorem},
  journal = {Physical Review Letters},
  volume = {67},
  number = {6},
  pages = {661--663},
  year = {1991},
  doi = {10.1103/PhysRevLett.67.661}
}

@article{Briegel1998,
  author = {Briegel, Hans-J. and D{\"u}r, Wolfgang and Cirac, J. Ignacio and Zoller, Peter},
  title = {Quantum Repeaters: The Role of Imperfect Local Operations in Quantum Communication},
  journal = {Physical Review Letters},
  volume = {81},
  number = {26},
  pages = {5932--5935},
  year = {1998},
  doi = {10.1103/PhysRevLett.81.5932}
}

@article{Cirac1999Distributed,
  author = {Cirac, J. Ignacio and Ekert, Artur K. and Huelga, Susana F. and Macchiavello, Chiara},
  title = {Distributed Quantum Computation over Noisy Channels},
  journal = {Physical Review A},
  volume = {59},
  number = {6},
  pages = {4249--4254},
  year = {1999},
  doi = {10.1103/PhysRevA.59.4249}
}

@article{ChitambarGour2019,
  author = {Chitambar, Eric and Gour, Gilad},
  title = {Quantum Resource Theories},
  journal = {Reviews of Modern Physics},
  volume = {91},
  number = {2},
  pages = {025001},
  year = {2019},
  doi = {10.1103/RevModPhys.91.025001}
}

@article{Regula2019OneShot,
  author = {Regula, Bartosz and Fang, Kun and Wang, Xin and Gu, Mile},
  title = {One-Shot Entanglement Distillation beyond Local Operations and Classical Communication},
  journal = {New Journal of Physics},
  volume = {21},
  number = {10},
  pages = {103017},
  year = {2019},
  doi = {10.1088/1367-2630/ab4732},
  eprint = {1906.01648},
  archivePrefix = {arXiv},
  primaryClass = {quant-ph}
}

@article{VandenbergheBoyd1996,
  author = {Vandenberghe, Lieven and Boyd, Stephen},
  title = {Semidefinite Programming},
  journal = {SIAM Review},
  volume = {38},
  number = {1},
  pages = {49--95},
  year = {1996},
  doi = {10.1137/1038003}
}

\newpage
\addtocontents{toc}{\protect\setcounter{tocdepth}{0}}
\appendix
\begin{center}
\LARGE\textbf{Appendix}
\end{center}

\section{A Regularization of the Exact PPT Distillation SDP}
\label{app:block-calculus}

\begin{lemma}[Continuous and integer-output regularization]
\label{lem:regularized-rate-exists}
For nonzero bipartite projections $P$ and $Q$,
\begin{align}
    W_0(P\ox Q)\le W_0(P)W_0(Q).
    \label{eq:W0-submultiplicative}
\end{align}
For a fixed nonzero $P$, define
\begin{align}
    W_n&\coloneqq W_0(P^{\ox n}),&
    a_n&\coloneqq-\log W_n,&
    b_n&\coloneqq\log\left\lfloor W_n^{-1}\right\rfloor .
    \label{eq:regularization-sequences}
\end{align}
Then both $(a_n)$ and $(b_n)$ are superadditive, their normalized
limits exist and are finite, and
\begin{align}
    \sup_{n\ge1}\frac{b_n}{n}
    =\lim_{n\to\infty}\frac{b_n}{n}
    =\lim_{n\to\infty}\frac{a_n}{n}
    =\sup_{n\ge1}\frac{a_n}{n}.
    \label{eq:W0-regularized-limit}
\end{align}
\end{lemma}

\begin{proof}
Let $E$ and $F$ be feasible for $W_0(P)$ and $W_0(Q)$,
respectively.  Since $P\le E\le\idop$ and $Q\le F\le\idop$,
\begin{align}
    E\ox F-P\ox Q
    &=(E-P)\ox F+P\ox(F-Q)\ge0,\\
    \idop\ox\idop-E\ox F
    &=(\idop-E)\ox\idop+E\ox(\idop-F)\ge0.
\end{align}
Thus $E\ox F$ is feasible for $W_0(P\ox Q)$.  Taking partial transpose on the $B$ system of each factor gives
\begin{align}
    \|(E\ox F)^\Gamma\|_\infty
    =\|E^\Gamma\ox F^\Gamma\|_\infty
    =\|E^\Gamma\|_\infty\|F^\Gamma\|_\infty.
\end{align}
Taking the infimum over $E$ and $F$ proves Eq.~\eqref{eq:W0-submultiplicative}.

We first record that all quantities in
Eq.~\eqref{eq:regularization-sequences} are finite.  Set
$d=\dim(\cH_{AB})$ and $r=\tr P\ge1$.  Every $E$ feasible for
$W_n$ satisfies
\begin{align}
    d^n\|E^\Gamma\|_\infty
    \ge|\tr E^\Gamma|
    =\tr E
    \ge\tr(P^{\ox n})
    =r^n.
\end{align}
The feasible effect $E=\idop$ gives the reverse trivial bound
$W_n\le1$.  Hence
\begin{align}
    0<\left(\frac rd\right)^n\le W_n\le1,
    \qquad
    0\le a_n\le n\log\frac dr.
    \label{eq:Wn-finite-bounds}
\end{align}
Applying Eq.~\eqref{eq:W0-submultiplicative} to $P^{\ox n}$ and $P^{\ox m}$ yields
$W_{n+m}\le W_nW_m$, and therefore $a_{n+m}\ge a_n+a_m$.  Next set $x_n=W_n^{-1}\ge1$.  The same
inequality gives $x_{n+m}\ge x_nx_m$, so
\begin{align}
    \lfloor x_{n+m}\rfloor
    \ge\lfloor x_nx_m\rfloor
    \ge\lfloor x_n\rfloor\lfloor x_m\rfloor.
    \label{eq:floor-supermultiplicative}
\end{align}
The last inequality holds because
$x_nx_m\ge\lfloor x_n\rfloor\lfloor x_m\rfloor$ and its
right-hand side is an integer.  Taking logarithms proves
$b_{n+m}\ge b_n+b_m$.  Fekete's lemma therefore applies to
both sequences.

Finally, we have that
\begin{align}
    \frac{x_n}{2}\le\lfloor x_n\rfloor\le x_n
\end{align}
holds for every $x_n\ge1$. This is immediate for $1\le x_n<2$, while for $x_n\ge2$ one has $\lfloor x_n\rfloor\ge x_n-1\ge x_n/2$.  Because all logarithms
are base two, this gives
\begin{align}
    a_n-1\le b_n\le a_n.
    \label{eq:integer-continuous-block-comparison}
\end{align}
Divide by $n$ and let $n\to\infty$.  The normalized limits coincide, and Fekete's lemma identifies each with the corresponding supremum, proving Eq.~\eqref{eq:W0-regularized-limit}.
\end{proof}

\section{Antisymmetric supports}\label{app:antisym}

Let $F$ be the swap on $\CC^d\ox\CC^d$ and define
\begin{align}
P_-\coloneqq\frac{\idop-F}{2},\qquad P_+\coloneqq\frac{\idop+F}{2}.
\end{align}
For $d\ge2$, note that $F^\Gamma=d\Phi_d$.  The effect
\begin{align}
E\coloneqq P_-+\frac{d-2}{d+2}P_+
\label{eq:antisymmetric-effect}
\end{align}
has coefficient $(d-2)/(d+2)$ on the symmetric subspace and coefficient one on
the antisymmetric subspace.  Since $0\le(d-2)/(d+2)\le1$ for $d\ge2$, this
proves $P_-\le E\le\idop$.  Rewrite $E$ and $E^\Gamma$ as
\begin{align}
E=\frac{d}{d+2}\idop-\frac{2}{d+2}F,
\qquad
E^\Gamma=\frac{d}{d+2}\idop-\frac{2d}{d+2}\Phi_d.
\end{align}
We see that
\begin{align}
\|E^\Gamma\|_\infty=\frac{d}{d+2}.
\label{eq:antisymmetric-upper}
\end{align}
Because $E$ is feasible, this proves
$W_0(P_-)\le d/(d+2)$.  For the converse bound, note that
\begin{align}
P_-^\Gamma=\frac12(\idop-d\Phi_d).
\end{align}
Its eigenvalue on the range of $\Phi_d$ is $-(d-1)/2$, while its eigenvalue on
the $(d^2-1)$-dimensional orthogonal complement is $1/2$.  Therefore
\begin{align}
\|P_-^\Gamma\|_1
=\frac{d-1}{2}+\frac{d^2-1}{2}
=\frac{(d-1)(d+2)}{2}.
\label{eq:antisymmetric-P-trace-norm}
\end{align}
Consequently,
\begin{align}
V\coloneqq\frac{2}{(d-1)(d+2)}P_-
\label{eq:antisymmetric-primal}
\end{align}
is positive and satisfies $\|V^\Gamma\|_1=1$, so it is primal feasible.  Since
$\tr P_-=d(d-1)/2$ and $P_-^2=P_-$, its objective value is
\begin{align}
\tr(P_-V)
=\frac{2\tr P_-}{(d-1)(d+2)}
=\frac{d}{d+2}.
\end{align}
Thus $M(P_-)\ge d/(d+2)$.  Combining this with the general inequality
$W_0(P_-)\ge M(P_-)$ and the effect upper bound gives
\begin{align}
W_0(P_-)=M(P_-)=\frac{d}{d+2}.
\end{align}
For $n$ copies, tensoring the feasible effect in
\eqref{eq:antisymmetric-effect} gives the upper bound
$W_0(P_-^{\ox n})\le[d/(d+2)]^n$.  Conversely,
the general ordering $W_0(P_-^{\ox n})\ge M(P_-^{\ox n})$ and
multiplicativity of $M$ give the same
quantity as a lower bound.  Therefore
\begin{align}
W_0(P_-^{\ox n})=M(P_-)^n
=\left(\frac{d}{d+2}\right)^n.
\label{eq:antisymmetric-all-n}
\end{align}
Here, equality is enforced by matching primal and effect certificates that already tensorize without loss.

\end{document}